\documentclass[conference]{IEEEtran}
\usepackage{etex}
\usepackage{graphicx}
\usepackage[latin1]{inputenc}
\usepackage[T1]{fontenc}
\usepackage[automark]{scrpage2}
\usepackage{amssymb}
\usepackage{amsmath}
\usepackage{algpseudocode}
\usepackage{algorithm}
\usepackage{verbatim}
\usepackage{dsfont}
\usepackage{amsthm}
\usepackage{psfrag}
\usepackage[off]{auto-pst-pdf} 
\usepackage{epsfig}
\usepackage{todonotes}
\usepackage{extarrows}
\usepackage{subfigure}
\usepackage{epstopdf}
\usepackage{cite}
\usepackage{multirow}

\usepackage{tikz}
\usetikzlibrary{arrows}
\tikzset{>=latex'}
\usepackage{pgfplots}

\ifCLASSINFOpdf
\newtheorem{Lemma}{Lemma}

\newtheorem{Prop}{Proposition}
\newtheorem{Theorem}{Theorem}

\newtheorem{Rem}{Remark}

\makeatletter
\newcommand{\Vast}{\bBigg@{4}}
\makeatother
\ifCLASSINFOpdf
 
\else
 
\fi

\definecolor{cn}{rgb}{0,0,0}
\definecolor{to_change}{rgb}{0,0,0}
\definecolor{comm_Ay}{rgb}{0,0,0}
\definecolor{Comm_An}{rgb}{0,0,0}
\definecolor{Comm_Ch}{rgb}{0,0,0}
\IEEEoverridecommandlockouts

\begin{document}

\title{Extended Generalized DoF Optimality Regime of Treating Interference as Noise in the X Channel}

\author{\IEEEauthorblockN{Soheil Gherekhloo, Anas Chaaban, and Aydin Sezgin}\thanks{This work was supported by the DFG grant SE 1697/10.
}
\IEEEauthorblockA{Institute of Digital Communication Systems\\
Ruhr-Universit\"at Bochum\\
Email: soheyl.gherekhloo,anas.chaaban,aydin.sezgin@rub.de}
}

\maketitle

\begin{abstract}
The simple scheme of treating interference as noise (TIN) is studied in this paper for the $3\times2$ X channel. A new sum-capacity upper bound is derived. This upper bound is transformed into a generalized degrees-of-freedom (GDoF) upper bound, and is shown to coincide with the achievable GDoF of a scheme that combines TDMA and TIN for some conditions on the channel parameters. These conditions specify a noisy interference regime which extends noisy interference regimes available in literature. As a by-product, the sum-capacity of the $3\times 2$ X channel is characterized within a constant gap in the given noisy interference regime.
\end{abstract}

\IEEEpeerreviewmaketitle

%\begin{IEEEkeywords}
%\end{IEEEkeywords}

\section{Introduction}

Dealing with interference is a main challenge in wireless communications. 
Compared with noise, interference contains information. Using this
property, some techniques have been investigated which decode the interference~\cite{HanKobayashi1981} 
in order to have a cleaner version of the received signal. This was  shown to be optimal in some cases such as 2-user interference channel (IC) with strong interference~\cite{Sato, Carleial_vsi}. On the other hand,
there is another extreme case in which the interference is so weak that
the undesired receiver is not able to decode it. Ignoring the interference
completely at the undesired receiver is a common way to deal with interference in such cases. 
This technique is known as treating interference as noise
(TIN)~\cite{CharafeddineSezginPaulraj}.

TIN is simple from a computational point of view, and is not demanding in terms of channel state information and coordination between different nodes. 
This simplicity makes TIN an appropriate choice for practical communication scenarios.
The practical advantages of TIN makes it interesting to identify cases where TIN is capacity-optimal~\cite{ShangKramerChen}. In~\cite{EtkinTseWang}, it is shown
that TIN is capacity optimal within a gap of $1$ bit in a 2-user IC which satisfies $\sqrt{\text{INR}}< \text{SNR}$. TIN is constant-gap optimal in the symmetric $K$-user IC ($K>2$) under the same condition~\cite{JafarVishwanath}. This result has been extended for asymmetric $K$-user IC in~\cite{GengNaderializadehAvestimehrJafar}. 
The optimality of TIN has also been studied for scenarios in which numbers of transmitters and receivers are not equal~\cite{ChaabanSezgin_SubOptTIN}. In such cases, it turns out that the constant-gap optimality regime of TIN is increased~\cite{GherekhlooDiChaabanSezgin2014} by switching some users off.
For a more general $M\times N$ X channel, conditions for \textcolor{comm_Ay}{constant gap optimality in the} noisy
interference regime were identified in~\cite{GengSunJafar2014}. \textcolor{comm_Ay}{The X channel models a cellular network, in which a user is communicating with multiple base stations in order to achieve a soft hand-over.}

In this paper, we show that the noisy interference regime of the X channel \textcolor{comm_Ay}{for which a constant gap to capacity is achieved} is in fact larger than the one given in~\cite{GengSunJafar2014}. To do this, we consider a $3\times2$ X channel for simplicity, and we derive a noisy interference regime where a TDMA-TIN scheme (which combines TDMA and TIN) is constant-gap optimal. The resulting noisy
interference regime identified in our work not only subsumes the
regime in~\cite{GengSunJafar2014}, but also extends it. This is mainly due to a novel
upper bound that we establish in this paper.

Throughout the paper, we use $C(x)=\log_2(1+x)$ for $x>0$, $\bar{x}=1-x$ for $x\in[0,1]$, and $x^n=(x_1,\cdots,x_n)$.

\section{System Model}
\label{Model}
The system we consider is a $3\times 2$ Gaussian X channel which consists of three senders and two receivers (Figure \ref{sysmod}). Each sender wants to communicate with each receiver. Namely, transmitter $i$ (Tx$i$) wants to send the messages $W_{ji}$ to receiver $j$ (Rx$j$), where $i\in\{1,2,3\}$ and $j\in\{1,2\}$. The message $W_{ji}$ has a rate $R_{ji}$. Tx$i$ encodes $(W_{1i},W_{2,i})$ into a codeword $x_i^n\in\mathbb{C}^n$ of $n$ symbols. The transmitters have power constraints $\rho$ which must be satisfied by their transmitted signals. 

\newcommand{\GaussianSystemmodel}[0]{
\node (t3) at (0,0) [inner sep=0] {};
\node (t1) at (0,-2) [inner sep=0] {};
\node (t2) at (0,-4) [inner sep=0] {};
\node (r1) at (3.5,-1) [inner sep=0] {};
\node (r2) at (3.5,-3) [inner sep=0] {};
\node (Z1U) at (3.6,-.3) [inner sep=0] {};
\node (Z1D) at (3.6,-.9) [inner sep=0] {};
\node (Z2D) at (3.6,-3.7) [inner sep=0] {};
\node (Z2U) at (3.6,-3.1) [inner sep=0] {};
\draw[->] (t3) to (r1);
\draw[->] (t1) to (r1);
\draw[->] (t2) to (r2);
\draw[->] (Z1U) to (Z1D);
\draw[->] (Z2D) to (Z2U);
\draw[->] (t3) to (r2);
\draw[->] (t1) to (r2);
\draw[->] (t2) to (r1);
\draw [fill=white, white] (0.6,-4) rectangle (1.3,0);
\node (h13) at (1,-0.3) [inner sep=0] {$h_{11}$};
\node (h23) at (1,-0.9) [inner sep=0] {{$h_{21}$}};
\node (h11) at (1,-1.7) [inner sep=0] {{$h_{12}$}};
\node (h22) at (1,-3.8) [inner sep=0] {{$h_{23}$}};
\node (h21) at (1,-2.3) [inner sep=0] {{$h_{22}$}};
\node (h22) at (1,-3.2) [inner sep=0] {{$h_{13}$}};
\node (tx3) at (-1.5,0) [inner sep=0] {\small{$\begin{pmatrix}W_{11}\\W_{21}\end{pmatrix}\rightarrow x_1^n$}};
\node (tx1) at (-1.5,-2) [inner sep=0] {\small{$\begin{pmatrix}W_{12}\\W_{22} \end{pmatrix}\rightarrow x_2^n$}};
\node (tx2) at (-1.5,-4) [inner sep=0] {\small{$\begin{pmatrix}
 W_{13} \\ W_{23}
\end{pmatrix}\rightarrow x_3^n$}};
\node (W1Rl) at (3.6,-1) [inner sep=0] {{$\oplus$}};
\node (W2Rl) at (3.6,-3) [inner sep=0] {{$\oplus$}};
\node (Z_1) at (3.75,0) [inner sep=0] {{$z_1^n$}};
\node (Z_2) at (3.75,-4) [inner sep=0] {{$z_2^n$}};
\node (Rx1) at (5.2,-1) [inner sep=0] {\small{  $\rightarrow y_1^n\rightarrow \begin{pmatrix}
\hat{W}_{11}\\ \hat{W}_{12}\\ \hat{W}_{13}\end{pmatrix}$}};
\node (Rx2) at (5.2,-3) [inner sep=0] {\small{$\rightarrow y_2^n\rightarrow \begin{pmatrix}
\hat{W}_{21}\\ \hat{W}_{22}\\ \hat{W}_{23}
\end{pmatrix}$}};
}
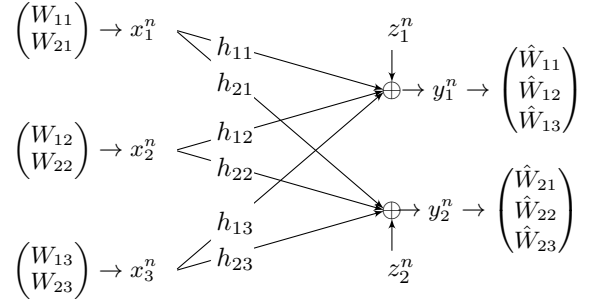
\begin{figure}[t]
\centering
\textcolor{comm_Ay}{
\begin{tikzpicture}[scale=0.8]
\GaussianSystemmodel
\end{tikzpicture}}
\caption{System model of the $3\times2$ Gaussian X channel.}
\label{sysmod}
\end{figure}   

At time instant $t\in\{1,\cdots,n\}$, Rx$j$ receives\footnote{The time index $t$ will be suppressed henceforth for clarity.}
\begin{align}
\label{recsig1}
y_j[t]&= h_{j1} x_1[t]+h_{j2} x_2[t]+h_{j3} x_3[t]+z_j[t],
\end{align}
where $z_j[t]$, $j\in\{1,2\}$, is a complex-valued Gaussian noise with zero mean and unit variance, and the constant $h_{ji}$ represents the complex (static) channel coefficient between Tx$i$ and Rx$j$. \textcolor{comm_Ay}{Since noise variance is unit, the transmit signal-to-noise ratio is given by $\rho$.} The noises $z_1$ and $z_2$ are independent of each other and are both independent and identically distributed (i.i.d.) over time. 

After $n$ transmissions, Rx$j$ has $y_j^n$ and decodes $W_{ji}$, $i\in\{1,2,3\}$. The probability of error, achievable rates $R_{ji}$, capacity region $\mathcal{C}$ are defined in the standard Shannon sense \cite{CoverThomas}. The sum-capacity $C_\Sigma$ is the maximum achievable sum-rate $R_\Sigma=\sum_{i=1}^3\sum_{j=1}^2R_{ji}$ for all rate tuples in the capacity region $\mathcal{C}$.
%which is given by
%\begin{align}\label{sumcap}
%{C}_{\Sigma}=\max_{(R_{11},R_{12},R_{21},R_{22},R_{31},R_{32})
%\in\mathcal{C}}R_\Sigma.
%\end{align}

In this work, 
%we assume that global channel state information is available to all nodes. 
we focus on the interference limited scenario, and hence, we assume that all signal-to-noise and interference-to-noise ratios are larger than 1, i.e., $\rho|h_{ji}|^2>1$ for all $j\in\{1,2\}$ and $i\in\{1,2,3\}$. Defining
\begin{align}
\alpha_{ji}=\frac{\log_2(\rho|h_{ji}|^2)}{\log_2(\rho)}, \label{eq:def_alpha}
\end{align} we define the generalized degrees-of-freedom (GDoF) of the channel as in \cite{GengNaderializadehAvestimehrJafar}
\begin{align}
\label{GDoFdef}
d_\Sigma(\boldsymbol{\alpha})=\lim_{\rho\to\infty}\frac{C_{\Sigma}(\boldsymbol{\alpha})}{\log_2(\rho)},
\end{align}
where $\boldsymbol{\alpha}$ is a vector which contains all $\alpha_{ji}$.

The focus of this work is studying constant gap optimality of TIN for the $3\times 2$ Gaussian X channel. Next, we introduce the transmission strategy we propose in this paper.

\section{TDMA-TIN}
\label{sec:TDMA-TIN}
In this scheme, we allow only two transmitters to be active simultaneously. In addition to this, for each active transmitter only one dedicated receiver is considered. \textcolor{comm_Ay}{Thus, we decompose the X channel into its underlying interference channels (IC).} In total, we have six  2-user IC's in the $3\times2$ X channel. Using TDMA, we assign a $\tau_s>0$ fraction of time to each of those six 2-user IC's, with $\sum_{s=1}^6 \tau_s = 1$. If the achievable sum-rate using TIN for one of those IC's is $R_s$, then the achievable sum-rate of TDMA-TIN is given by 
\begin{align}
R_{TT}= \max_{\tau_1,\ldots,\tau_6}\sum_{s=1}^6 \tau_s R_s.
\end{align}
This optimization problem is linear in $\tau_s$ and is solved by setting $\tau_s=1$ for some $s\in\{1,\cdots, 6\}$ and setting the remaining $\tau_{s'}=0$. Namely, the maximization above is achieved by activating the 2-user IC which yields the highest sum-rate. Without loss of generality, suppose that the 2-user IC with maximum TIN sum-rate is the one in which Tx$i_1$ and Tx$i_2$ want to send messages $W_{j_1i_1}$ and $W_{j_2i_2}$ to Rx$j_1$ and Rx$j_2$, respectively. The transmitters encode their message into a codeword with power $\rho$. This causes interference at undesired receivers. Therefore, the receivers decode their desired messages using TIN. Using this scheme, the following sum-rate is achievable
\begin{align*}
R_{j_1i_1} + R_{j_2i_2} = C\left(\frac{\rho^{\alpha_{j_1i_1}}}{1+\rho^{\alpha_{j_1i_2}}}\right)+ C\left(\frac{\rho^{\alpha_{j_2i_2}}}{1+\rho^{\alpha_{j_2i_1}}}\right) .
\end{align*}
In general, the achievable sum-rate by using TDMA-TIN is presented
in the following proposition.

\begin{Prop}
\label{Pro:Ach_TDMA-TIN_Gaussian_32X}
The achievable sum-rate of TDMA-TIN in the $3\times 2$ Gaussian X channel is given by
\begin{align}
\label{eq:GausAchTDMATIN}
R_{TT}  =  &\max_{\boldsymbol{p}_{TT}} R_{TT}(\boldsymbol{p}_{TT})
\end{align}
where $\boldsymbol{p}_{TT}=(i_1,i_2,j_1,j_2)$ with $i_1,i_2 \in \{1,2,3\}$, $j_1,j_2\in\{1,2\}$, $i_1\neq i_2$, $j_1\neq j_2$, and where $$R_{TT}(\boldsymbol{p}_{TT})=C\left(\frac{\rho^{\alpha_{j_1i_1}}}{1+\rho^{\alpha_{j_1i_2}}}\right) + C\left(\frac{\rho^{\alpha_{j_2i_2}}}{1+\rho^{\alpha_{j_2i_1}}}\right).$$
\end{Prop}
Let us transform this achievable rate expression to an achievable GDoF expression. We first bound $R_{TT}(\boldsymbol{p}_{TT})$ as follows
\begin{align*}
&R_{TT}(\boldsymbol{p}_{TT})\\
&>  \left[(\alpha_{j_1i_1}-\alpha_{j_1i_2})^+
+(\alpha_{j_2i_2}-\alpha_{j_2i_1})^+\right]
\log_2(\rho)-2.
\end{align*}
Therefore, for this particular set of transmitters and receivers,  TDMA-TIN achieves a GDoF of 
\begin{align*}
D_{TT}(\boldsymbol{p}_{TT})=(\alpha_{j_1i_1}-\alpha_{j_1i_2})^+
+(\alpha_{j_2i_2}-\alpha_{j_2i_1})^+.
\end{align*}
As a result, TDMA-TIN achieves the following GDoF
\begin{align}
\label{eq:GDoFTT}
d_{TT}(\boldsymbol{\alpha}) =  &\max_{\boldsymbol{p}_{TT}} D_{TT}(\boldsymbol{p}_{TT}).
\end{align}
Despite the simplicity of TDMA-TIN, this scheme is constant-gap optimal in some cases as we shall see next.

\section{Constant-gap optimality of TDMA-TIN}
Here, we want to introduce a noisy interference regime in which TDMA-TIN achieves the GDoF of the $3\times 2$ Gaussian X channel, and \textcolor{comm_Ay}{moreover}, achieves its sum-capacity within a constant gap. The following theorem characterizes the GDoF of the channel in such a noisy interference regime.
\begin{Theorem}
\label{Th:TIN_con_gap_32X_Gauss}
If there exist distinct $i_1,i_2,i_3\in\{1,2,3\}$ and distinct $j_1,j_2\in\{1,2\}$ such that the following noisy interference regime conditions are satisfied:
\begin{align}
\label{eq:condition_Opt_TDMA_TIN_Gauss1}
\alpha_{j_1i_1}-\alpha_{j_2i_1}&\geq \psi\\
\label{eq:condition_Opt_TDMA_TIN_Gauss2}
\alpha_{j_2i_2}-\alpha_{j_1i_2}&\geq \max\{\alpha_{j_2i_1},\alpha_{j_2i_3}\},
\end{align}
where $\psi=\max\{\alpha_{j_1i_3}-(\alpha_{j_2i_3}-\alpha_{j_2i_1})^+,\alpha_{j_1i_2}\}$, then TDMA-TIN achieves the GDoF of the $3\times2$ Gaussian X channel given by
\begin{align}
\label{Thm:GDoFExp}
d_\Sigma(\boldsymbol{\alpha})&\leq \alpha_{j_1i_1}-\alpha_{j_2i_1}+\alpha_{j_2i_2}-\alpha_{j_1i_2}.
\end{align}
\end{Theorem}
In other words, if there exists a permutation of the transmitters and receivers such that conditions \eqref{eq:condition_Opt_TDMA_TIN_Gauss1} and \eqref{eq:condition_Opt_TDMA_TIN_Gauss2} hold, then TDMA-TIN is GDoF-optimal. Using this theorem, we can show that TDMA-TIN achieves the sum-capacity of the channel within a constant gap. It can be shown that this gap can be upper bounded by 7 bits as long as the conditions in \eqref{eq:condition_Opt_TDMA_TIN_Gauss1} and  \eqref{eq:condition_Opt_TDMA_TIN_Gauss2} are satisfied. Due to space limitations, the gap analysis is not included.

Note that conditions \eqref{eq:condition_Opt_TDMA_TIN_Gauss1} and \eqref{eq:condition_Opt_TDMA_TIN_Gauss2} specify a larger noisy interference regime than that identified in \cite{GengSunJafar2014}. This is true since $\psi=\max\{\alpha_{j_1i_3}-(\alpha_{j_2i_3}-\alpha_{j_2i_1})^+,\alpha_{j_1i_2}\}$ is smaller than $\max\{\alpha_{j_1i_3},\alpha_{j_1i_2}\}$ as identified in \cite{GengSunJafar2014}, specifically, if $\alpha_{j_2i_3}-\alpha_{j_2i_1}>0$. \textcolor{Comm_An}{In Fig.~\ref{fig:GDoF_opt_TIN}, the GDoF optimality regime of TDMA-TIN is illustrated for $3\times 2$ X channel with $\alpha_{11} = \alpha_{22}=1$, $\alpha_{13}= \alpha_{23}=\beta$, where $\beta$ is larger than $0.5$ and smaller than $1$.
The noisy interference regime obtained from Theorem \ref{Th:TIN_con_gap_32X_Gauss} is given by the union of the rectangle defined by $(\alpha_{21},\alpha_{12})\in[0, 0.5]\times[0,1-\beta]$ and the rectangle defined by $(\alpha_{21},\alpha_{12})\in[0,1-\beta]\times[0, 0.5]$. On the other hand, the noisy interference regime obtained from \cite{GengSunJafar2014} is given by the intersection of these two rectangles. 
Obviously, the new noisy interference regime does not only subsume the previously known regime from \cite{GengSunJafar2014} but also extends it. Note that if the channel gains from Tx$i_3$ to the receivers decrease, then the intersection of the two rectangles increases. At the point $\beta=1/2$, both regimes will coincide and the noisy interference regime becomes a rectangle with width and height of $1/2$.}
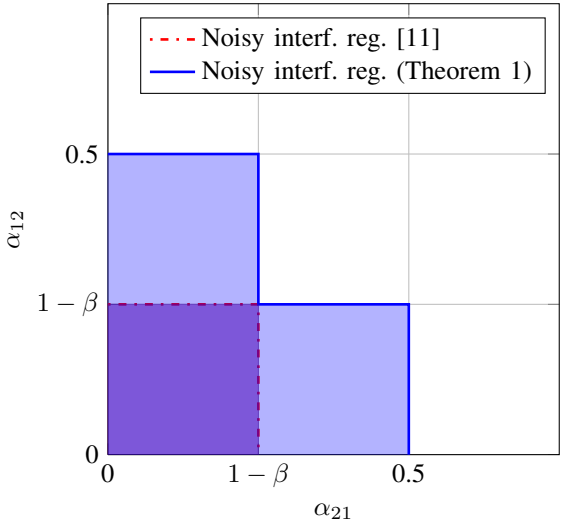
\begin{figure}
\begin{tikzpicture}

\begin{axis}[%
width=6cm,
height=6cm,
scale only axis,
xmin=0,
xmax=0.75,
xtick={0,0.25,0.5},
xticklabels={{0},{$1-\beta$},{0.5}},
xlabel={$\alpha_{21}$},
xmajorgrids,
ymin=0,
ymax=0.75,
ytick={0,0.25,0.5},
yticklabels={{0},{$1-\beta$},{0.5}},
ylabel={$\alpha_{12}$},
ymajorgrids,
legend style={at={(0.97,0.8)},anchor=south east,draw=black,fill=white,legend cell align=left}
]
\addplot [color=red,dash pattern=on 1pt off 3pt on 3pt off 3pt,line width=1pt]
  table[row sep=crcr]{0	0.25\\
0.25	0.25\\
0.25	 0\\
};
\addlegendentry{Noisy interf. reg. \cite{GengSunJafar2014}};

\addplot [color=blue,solid,line width=1.0pt]
  table[row sep=crcr]{0	0.5\\
0.25 	0.5\\
0.25	 0.25\\
0.5 0.25\\
0.5 0.0\\
};
\addlegendentry{Noisy interf. reg. (Theorem \ref{Th:TIN_con_gap_32X_Gauss})};

\addplot[area legend,solid,fill=red,opacity=0.3,draw=red,forget plot]
table[row sep=crcr] {%
x	y\\
0	0\\
0	0.25\\
0.25	 0.25\\
0.25	 0\\
};

\addplot[area legend,solid,fill=blue,opacity=0.3,draw=blue,forget plot]
table[row sep=crcr] {%
x	y\\
0	0\\
0	0.5\\
0.25	0.5\\
0.25	0\\
};
\addplot[area legend,solid,fill=blue,opacity=0.3,draw=blue,forget plot]
table[row sep=crcr] {%
x	y\\
0	0\\
0.5	0\\
0.5	0.25\\
0	0.25\\
};

\end{axis}
\end{tikzpicture}%
\caption{The GDoF optimality regime of TDMA-TIN for $3\times 2$ X channel with $\alpha_{11} = \alpha_{22}=1$, and $\alpha_{13}=\alpha_{23}=\beta$, where $0.5<\beta<1$.}
\label{fig:GDoF_opt_TIN}
\end{figure}
%The resulting gap is given in the following corollary.
%\begin{Cor}
%\label{Cor:Gap}
%The gap between the achievable sum-rate of TDMA-TIN and the sum-capacity of the Gaussian $3\times2$ X channel is bounded by 7 bits as long as the conditions in \eqref{eq:condition_Opt_TDMA_TIN_Gauss1} and  \eqref{eq:condition_Opt_TDMA_TIN_Gauss2} are satisfied.
%\end{Cor}

The GDoF expression given in Theorem \ref{Th:TIN_con_gap_32X_Gauss} is clearly achievable by TDMA-TIN. Namely, consider a permutation of transmitters and receivers given by distinct $i_1,i_2,i_3\in\{1,2,3\}$ and distinct $j_1,j_2\in\{1,2\}$. Then, \eqref{eq:GDoFTT} leads to
\begin{align}
d_{TT}(\boldsymbol{\alpha}) &=  \max_{\boldsymbol{p}_{TT}} D_{TT}(\boldsymbol{p}_{TT}) >  D_{TT}(i_1,i_2,j_1,j_2)\\
%&>  D_{TT}(i_1,i_2,j_1,j_2)\\
&= \alpha_{j_1i_1}-\alpha_{j_1i_2}+\alpha_{j_2i_2}
-\alpha_{j_2i_1}
\end{align}
where the last step follows since the conditions in Theorem \ref{Th:TIN_con_gap_32X_Gauss} dictate that $\alpha_{j_1j_1}\geq\alpha_{j_1i_2}$ and $\alpha_{j_2j_2}\geq\alpha_{j_2i_1}$. This achievable GDoF coincides with \eqref{Thm:GDoFExp}. 

To prove Theorem \ref{Th:TIN_con_gap_32X_Gauss}, it remains to prove the converse. In other words, we still need to establish an upper bound on the GDoF which coincides with \eqref{Thm:GDoFExp} under the conditions \eqref{eq:condition_Opt_TDMA_TIN_Gauss1} and \eqref{eq:condition_Opt_TDMA_TIN_Gauss2}. The converse is provided in the next section.

\section{Converse \textcolor{comm_Ay}{for Theorem \ref{Th:TIN_con_gap_32X_Gauss}}}
Here, we derive an upper bound on the sum-capacity of the $3\times 2$ Gaussian X channel which proves the converse of Theorem \ref{Th:TIN_con_gap_32X_Gauss}. The upper bound is given in the following lemma.
\begin{Lemma}
\label{Th:UB_Det_32X}
The sum-capacity of the $3\times 2$ Gaussian X channel is upper bounded by
\begin{align}
\label{eq:UB_32X}
C_\Sigma\leq &\min_{\boldsymbol{p}} B(\boldsymbol{p})
\end{align}
where $\boldsymbol{p}=(i_1,i_2,i_3,j_1,j_2)$ for distinct $i_1,i_2,i_3 \in\{1,2,3\}$ and distinct $j_1,j_2\in\{1,2\}$, and where $B(\boldsymbol{p})$ is as given in \eqref{eq:Bij} \textcolor{comm_Ay}{on the top of the next page}.
\begin{figure*}
\begin{align}
\label{eq:Bij}
&B(\boldsymbol{p})= C\left(\rho^{\alpha_{j_1i_2}} + \bar{d} \rho^{\alpha_{j_1i_3}}+ \frac{\rho^{\alpha_{j_1i_1}} + d \rho^{\alpha_{j_1i_3}}}{1+c^2(\rho^{\alpha_{j_1i_1}} +d \rho^{\alpha_{j_1i_3}})}\right) + C\left(\rho^{\alpha_{j_2i_1}} + \rho^{\alpha_{j_2i_3}} +\frac{\rho^{\alpha_{j_2i_2}} }{1+\rho^{\alpha_{j_1i_2}}}\right)+1\\
&(c^2,d) = \begin{cases}
(\rho^{\alpha_{j_2i_1}-\alpha_{j_1i_1}}, 0)& \text{if}\quad  \alpha_{j_2i_3}\leq \alpha_{j_2i_1},
\\
(\rho^{\alpha_{j_2i_1}-\alpha_{j_1i_1}}, 1)& \text{if}\quad  \alpha_{j_2i_3}> \alpha_{j_2i_1}\text{ and } 
\alpha_{j_2i_1} - \alpha_{j_1i_1} \leq \alpha_{j_2i_3}-\alpha_{j_1i_3} - \alpha_{j_2i_1}, \\
(\rho^{\alpha_{j_2i_3}-\alpha_{j_2i_1}-\alpha_{j_1i_3}},1)& \text{otherwise.} \end{cases}
\label{eq:param_def_c_d}
\end{align}
\hrule
\end{figure*}
\end{Lemma}
\begin{proof}
%Fist, we show that for a fixed permutation of transmitters and receivers: $\boldsymbol{p}=(i_1,i_2,i_3,j_1,j_2)=(1,2,3,1,2)$, $C_\Sigma\leq B(\boldsymbol{p})$. Similarly, the upper bound for other permutations can be established. Getting the minimum of the upper bounds over all permutations, we obtain the upper bound in Lemma~\ref{Th:UB_Det_32X}.
We consider the following permutation of transmitters and receivers: $\boldsymbol{p}=(i_1,i_2,i_3,j_1,j_2)=(1,2,3,1,2)$; the other cases can be proved similarly. 
We give \textcolor{comm_Ay}{$V_{j_1}=V_1 = \{W_{21},W_{23}\}$} and $S_{j_1}^n=S_1^n=c(h_{11}X_1^n+ d\cdot h_{13}X_3^n)+N_1^n$ to Rx1 as side information,\footnote{The capital letter notation is used for random variables.} where\textcolor{comm_Ay}{
\begin{align}
(c,d) = 
\begin{cases}
(\frac{h_{21}}{h_{11}},0)                 & \frac{|h_{23}|}{|h_{21}|}\leq 1 \\
(\frac{h_{21}}{h_{11}},1)                 & \frac{|h_{23}|}{|h_{21}|}> 1,\,  \frac{\rho|h_{21}|^4}{|h_{11}|^2}\leq \frac{|h_{23}|^2}{|h_{13}|^2}\\
(\frac{h_{23}}{h_{21}\sqrt{\rho}h_{13}},1) & \frac{|h_{23}|}{|h_{21}|}> 1,\, \frac{\rho|h_{21}|^4}{|h_{11}|^2}> \frac{|h_{23}|^2}{|h_{13}|^2},
\end{cases}\label{eq:c_param}
\end{align}}
and where $N_1^n$ is a zero-mean unit-variance Gaussian noise independent of $Z_1$ and $Z_2$ and i.i.d. over time. 
%where \\
%\noindent{ Case 1:} $|h_{23}|\leq |h_{21}|$,\\
%\noindent{ Case 2:} $|h_{23}|> |h_{21}|$, and $\frac{h_{21}}{h_{11}}\leq \frac{h_{23}}{h_{13}\sqrt{\rho} h_{21}}$, \\
%\noindent{ Case 3:} $|h_{23}|> |h_{21}|$, and $\frac{h_{21}}{h_{11}} > \frac{h_{23}}{h_{13}\sqrt{\rho} h_{21}}$.

%where $c=\frac{h_{23}}{h_{21}\sqrt{\rho}h_{13}}$ if the channel satisfies $|h_{23}|^2>|h_{21}|^2$ and $\frac{\rho|h_{13}|^2|h_{21}|^2}{|h_{23}|^2} >\frac{|h_{11}|^2}{|h_{21}|^2}$, and $c=\frac{h_{21}}{h_{11}}$ otherwise, and $d=0$ if $|h_{23}|^2\leq|h_{21}|^2$, and $d=1$ otherwise, and where $N_1^n$ is a zero-mean unit-variance Gaussian noise independent of $Z_1$ and $Z_2$ and i.i.d. over time.

We also give $V_{j_2}=V_2=W_{12}$ and $S_{j_2}^n=S_2^n = h_{12}X_2^n+N_2^n$ to Rx2 as side information, where $N_2^n=Z_1^n$. 
%$$(c,d) = \begin{cases}(\frac{h_{21}}{h_{11}}, 0)&\quad \text{if}\quad \frac{P|h_{13}|^2|h_{21}|^2}{|h_{23}|^2} \leq\frac{|h_{11}|^2}{|h_{21}|^2} \text{ or } |h_{23}|^2\leq|h_{21}|^2
%\\
%(\frac{h_{23}}{h_{21}\sqrt{P}h_{13}},1)& \quad \text{if}\quad \frac{P|h_{13}|^2|h_{21}|^2}{|h_{23}|^2} >\frac{|h_{11}|^2}{|h_{21}|^2} \text{ and } |h_{23}|^2>|h_{21}|^2 \end{cases}$$
Using Fano's inequality, and defining $\tilde{W}_j$ as the set of messages desired at Rx$j$, i.e. $\{W_{j1},W_{j2},W_{j3}\}$, we obtain 
\begin{align}
nR_\Sigma\leq & \sum_{j=1}^2 I(\tilde{W}_j;Y_j^n,S_j^n,V_j)+n\epsilon_n, \notag 
\end{align}
where $\epsilon_n\to 0$ as $n\to\infty$. Then, using the chain rule, and since all messages are independent, we can write 
\begin{align*}
nR_\Sigma&\leq  \sum_{j=1}^2 \left[I(\tilde{W}_j;S_j^n|V_j) + I(\tilde{W}_j;Y_j^n|S_j^n,V_j)\right]+n\epsilon_n.
\end{align*}
Now by using $h(S_j^n|\tilde{W}_j,V_j)=h(N_j^n)$ we get
\begin{align}
nR_\Sigma&\leq
%&=\sum_{j=1}^2 \left[h(S_j^n|V_j) - h(S_j^n|\tilde{W}_j,V_j)+ h(Y_j^n|S_j^n,V_j)\right.\nonumber\\
%&\qquad \left.- h(Y_j^n|S_j^n,\tilde{W}_j,V_j)\right] \notag \\
 \sum_{j=1}^2 \left[h(S_j^n|V_j) - h(N_j^n)+ h(Y_j^n|S_j^n,V_j)\right.\nonumber\\
&\qquad\qquad \left.- h(Y_j^n|\tilde{W}_j,V_j)\right]+n\epsilon_n\label{RsSum4}
\end{align}
Now, defining $\tilde{S}_1^n = h_{21}X_1^n + h_{23} X_3^n + Z_{2}^n$ and $\tilde{S}_2^n = h_{12}X_2^n + Z_{1}^n$ %$\tilde{S}_j^n = S_j^n-N_j^n+Z_{j^\prime}^n$,where $j^\prime\in\{1,2\}$, $j^\prime\neq j$
and using the fact that $(X_1^n,X_3^n)$ and $X_2^n$ can be reconstructed from $(\tilde{W_1},V_1)$ and $(\tilde{W_2},V_2)$, respectively, we obtain $h(Y_j^n|\tilde{W}_j,V_j)=h(\tilde{S}_i^n|V_i)$, where $i\neq j$, $i,j\in\{1,2\}$. Furthermore, since conditioning does not increase entropy, we have $h(\tilde{S}_1^n|V_1)\geq h(\tilde{S}_1^n|V_1,\bar{d}X_3^n)$. Substituting in \eqref{RsSum4} yields
\begin{align}
nR_\Sigma&\leq \sum_{j=1}^2\left[h(S_j^n|V_j) - h(N_j^n)+ h(Y_j^n|S_j^n,V_j)\right]  \nonumber\\
&\quad - h(\tilde{S}_{2}^n|V_{2}) - h(\tilde{S}_{1}^n|V_{1},\bar{d} X_3^n)+n\epsilon_n\notag \\ 
\label{Rsigm4t}
&\overset{(c)}{\leq} \sum_{j=1}^2 \left[h(Y_j^n|S_j^n,V_j)- h(N_j^n)\right]+ n +n\epsilon_n
\end{align}
where $(c)$ follows since $h(S_2^n|V_2) = h(\tilde{S}_2^n|V_2)$, and since  $h(S_1^n|V_1) = h(\tilde{S}_1^n|V_1,X_3^n)$ if $d=0$, and $h(S_1^n|V_1) - h(\tilde{S}_1^n|V_1)\leq n$ if $d=1$ as shown in Lemma~\ref{Lemma:Entropy_diff_Gaussian} in Appendix~\ref{Proof_Gaus_UB_32X}. By dropping the conditioning on $V_1$ and $V_2$, using Lemma~1 in~\cite{AnnapureddyVeeravalli} which shows that a circularly symmetric complex Gaussian distribution maximizes the
conditional differential entropy for a given covariance constraint,  %and defining $P_u=P|h_{11}|^2+d P|h_{13}|^2$, we obtain 
%\begin{align}
%&h(Y_1^n|S_1^n,V_1)-h(N_1^n)\nonumber\\
%&\quad\leq nC\left(P|h_{12}|^2 + \bar{d} P|h_{13}|^2+ \frac{P_u}{1+c^2P_u}\right),\\
%&h(Y_2^n|S_2^n,V_2)-h(N_2^n)\nonumber\\
%&\quad\leq nC\left(P|h_{21}|^2 + P|h_{23}|^2 +\frac{P|h_{22}|^2 }{1+P|h_{12}|^2}\right).
%\end{align}
%Consequently, by substituting in \eqref{Rsigm4t}, 
dividing by $n$, letting $n\to\infty$, and using \eqref{eq:def_alpha}, we obtain
\begin{align}
R_\Sigma
&\leq 
C\left(\rho^{\alpha_{12}} + \bar{d} \rho^{\alpha_{13}}+ \frac{\rho^{\alpha_{11}} + d \rho^{\alpha_{13}}}{1+c^2(\rho^{\alpha_{11}} +d  \rho^{\alpha_{13}})}\right) \notag\\ &\quad + C\left(\rho^{\alpha_{21}} + \rho^{\alpha_{23}} +\frac{\rho^{\alpha_{22}} }{1+\rho^{\alpha_{12}}}\right)+1 \label{eq:UB_32X_end_proof}
\end{align}
which is equal to the desired bound $B(\boldsymbol{p})$ \eqref{eq:Bij} for this specific permutation $\boldsymbol{p}=(i_1,i_2,i_3,j_1,j_2)=(1,2,3,1,2)$ of transmitters and receivers. By rewriting the parameters $c$, $d$ as a function of $\rho$, we obtain \eqref{eq:param_def_c_d} for this permutation.  
%
%
%$$(c^2,d,d^\prime) = \begin{cases}(\rho^{r(m_{21}-m_{11})}, 0,1)&\quad \text{if}\quad  m_{23}\leq m_{21}
%\\
%(\rho^{r(m_{21}-m_{11})}, 1,0)&\quad \text{if}\quad  m_{21} - m_{11} \leq m_{23}-m_{13} - m_{21} \text{ and } m_{23}> m_{21}
%\\
%(\rho^{r(m_{23}-m_{21}-m_{13})},1,0)&\quad \text{if}\quad  m_{21} - m_{11} > m_{23}-m_{13} - m_{21} \text{ and } m_{23}>m_{21} \end{cases}$$
Writing the upper bound in \eqref{eq:UB_32X_end_proof} for all permutations of $i_1,i_2,i_3\in\{1,2,3\}$ and $j_1,j_2\in\{1,2\}$, we obtain the upper bound in \eqref{eq:UB_32X}.
\end{proof}

With this, we obtain a sum-capacity upper bound. We can use the definition of the GDoF in \eqref{GDoFdef} to write this upper bound as a GDoF upper bound. For this purpose, we divide $B(\boldsymbol{p})$ by $\log(\rho)$ and we let $\rho\to\infty$ to obtain a GDoF upper bound for each of the cases in \eqref{eq:param_def_c_d}. By combining the resulting GDoF upper bounds, we get the GDoF upper bound for a specific permutation $\boldsymbol{p}$ (details can be found in Appendix \ref{app:Constant Gap Analysis for TDMA-TIN})
\begin{align}
\label{GDoFUpperBound}
&D(\boldsymbol{p})<\max\{\alpha_{j_2i_1},\alpha_{j_2i_3},\alpha_{j_2i_2}-\alpha_{j_1i_2}\}\\
&+ \max\{\alpha_{j_1i_2},\alpha_{j_1i_1}-\alpha_{j_2i_1},\alpha_{j_1i_3}-(\alpha_{j_2i_3}-\alpha_{j_2i_1})^+\}
\nonumber.
\end{align}
Therefore $d_\Sigma(\boldsymbol{\alpha})\leq \min_{\boldsymbol{p}} D(\boldsymbol{p}).$ 
%\begin{align}
%d_\Sigma(\boldsymbol{\alpha})\leq \min_{\boldsymbol{p}} D(\boldsymbol{p}).
%\end{align}
Now, we have a general GDoF upper bound. Let us specialize this bound to the noisy interference regime of Theorem \ref{Th:TIN_con_gap_32X_Gauss}. Suppose that the conditions in Theorem \ref{Th:TIN_con_gap_32X_Gauss} are satisfied for some permutation of transmitters and receivers $\hat{\boldsymbol{p}}=(t_1,t_2,t_3,r_1,r_2)$. By using these conditions we get 
\begin{align*}
d_\Sigma(\boldsymbol{\alpha})&\leq \min_{\boldsymbol{p}} D(\boldsymbol{p})\nonumber\\
&\leq D(\hat{\boldsymbol{p}})=\alpha_{r_1t_1}-\alpha_{r_2t_1}+\alpha_{r_2t_2}-\alpha_{r_1t_2},
\end{align*}
which proves the converse of Theorem \ref{Th:TIN_con_gap_32X_Gauss}.

%Now, we are ready to prove Theorem \ref{Th:TIN_con_gap_32X_Gauss}. To this end, we apply the conditions in \eqref{eq:condition_Opt_TDMA_TIN_Gauss1} and \eqref{eq:condition_Opt_TDMA_TIN_Gauss2} to the achievable sum-rate and the upper bound. Then, we obtain the gap between them.
%It is shown in Appendix \ref{app:Constant Gap Analysis for TDMA-TIN} that the upper bound in \eqref{eq:UB_32X} can be relaxed as follow 
%\begin{align}
%R_\Sigma\leq&5+r\log_2 \rho (m_{j_1i_1}-m_{j_2i_1} + m_{j_2i_2}-m_{j_1i_2}), \label{eq:UB_3x2_G9}
%\end{align}
%if there exists $(i_1,i_2,i_3,j_1,j_2)$ which satisfies \eqref{eq:condition_Opt_TDMA_TIN_Gauss1} and \eqref{eq:condition_Opt_TDMA_TIN_Gauss2}. Under the same condition, it is also shown in Appendix \ref{app:LB_Achievable_sum_rate} that the achievable sum-rate in \eqref{eq:GausAchTDMATIN} can be bounded by 
%\begin{align}
%R_{\Sigma}  &> r\log_2 \rho [m_{j_1i_1}-m_{j_1i_2} + m_{j_2i_2}-m_{j_2i_1}] - 2.\label{eq:Ach_TDMA_TIN_Gauss_2}
%\end{align}
%By comparing \eqref{eq:UB_3x2_G9} with \eqref{eq:Ach_TDMA_TIN_Gauss_2}, we see that TDMA-TIN is within a gap of $7$ bits to the sum-capacity under the conditions of Theorem \ref{Th:TIN_con_gap_32X_Gauss}.

\appendices

\section{}
\label{Proof_Gaus_UB_32X}
In this appendix, we introduce a lemma which is necessary for proving the bound \eqref{eq:UB_32X}. Let $W_A$ and $W_B$ be two independent messages, and let $X_A$ (independent of $W_B$) and $X_B$ (independent of $W_A$) be two independent complex-valued signals satisfying a power constraint $\rho$. Define $Y_A$ and $Y_B$ as noisy channel outputs given by
 \begin{align}
   Y_A & = h_1 X_A + h_2 X_B+Z_A\\
   Y_B & = h_3 X_A + h_4 X_B+Z_B,
   \end{align}
where $Z_A$ and $Z_B$ are zero-mean unit-variance Gaussian noises, and are independent of each other and of all other random variables, and where the constants $h_1$, $h_2$, $h_3$ and $h_4$ are complex-valued and satisfy 
\begin{align}
|h_1|^2  \leq  |h_3|^2  \leq \frac{|h_4|^2}{\rho |h_2|^2}\quad \text{and} \quad 1  < \rho |h_3|^2 \label{eq:ExtCondition3}.
\end{align} 
Let $Y_A^{n}$ and $Y_B^{n}$ be the outputs corresponding to inputs $X_A^{n}$ and $X_B^{n}$ of length $n$, and define $W_C=(W_A,W_B)$. Then,  we have the following lemma.
\begin{Lemma}
\label{Lemma:Entropy_diff_Gaussian}
If conditions \eqref{eq:ExtCondition3} are satisfied, then we have
\begin{align}
\label{lem2}
h(Y_A^n|W_C)-h(Y_B^n|W_C)\leq n.
\end{align}
\end{Lemma}
\begin{Rem}
The parameter $c$ in \eqref{eq:c_param} is chosen such that the conditions in \textcolor{comm_Ay}{\eqref{eq:ExtCondition3}} are satisfied, and therefore in converse for Theorem \ref{Th:TIN_con_gap_32X_Gauss}, $h(S_1^n|V_1) - h(\tilde{S}_1^n|V_1)\leq n$. 
\end{Rem}
\begin{proof}
We start by upper bounding the difference as follows
\begin{align*}
h(Y_A^n&|W_C)-h(Y_B^n|W_C)\nonumber\\
&=I(X_A^n,X_B^n;Y_A^n|W_C) - I(X_A^n,X_B^n;Y_B^n|W_C)\\
&\overset{(a)}{\leq} I(X_A^n,X_B^n;Y_A^n|W_C) - I(X_A^n,X_B^n;Y_B^n|W_C)\nonumber\\
&\quad + I(X_A^n;X_B^n|Y_A^n,W_C)\\
&\overset{(b)}{\leq}I(X_A^n;Y_A^n,X_B^n|W_C) + I(X_B^n;Y_A^n|X_A^n,W_C)\nonumber\\
&\quad - I(X_B^n;Y_B^n|W_C) - I(X_A^n;Y_B^n|X_B^n,W_C),
\end{align*}
where $(a)$ follows from the non-negativity of mutual information and $(b)$ follows by using chain rule. Note that $I(X_A^n;X_B^n|W_C)=0$, and hence, $I(X_A^n;Y_A^n,X_B^n|W_C)=I(X_A^n;Y_A^n|X_B^n,W_C)$. Using some standard steps, we get
\begin{align}
&h(Y_A^n|W_C)-h(Y_B^n|W_C)  \nonumber\\
&\leq I(X_A^n;h_1 X_A^n + Z_A^n|W_C) + I(X_B^n;h_2 X_B^n+Z_A^n|W_C) \notag \\ &\quad- I(X_B^n;Y_B^n|W_C) - I(X_A^n; h_1 X_A^n +\frac{h_1}{h_3}Z_B^n|W_C)\nonumber\\
&\stackrel{(c)}{\leq} I(X_B^n;h_2 X_B^n+Z_A^n|W_C) - I(X_B^n;Y_B^n|W_C),\label{StB}
\end{align}
where $(c)$ follows since $I(X_A^n;h_1 X_A^n + Z_A^n|W_C)\leq I(X_A^n; h_1 X_A^n +\frac{h_1}{h_3}Z_B^n|W_C)$ since $Z_A$ and $Z_B$ have the same distribution and $|h_1|^2\leq |h_3|^2$. Next, we proceed by bounding $T=I(X_B^n;Y_B^n|W_C)$. First, we write 
\begin{align}
T&= I(X_B^n; \tilde{X}^n_A + \tilde{X}^n_B+\tilde{Z}_B^n|W_C)
\end{align}
where we define $\tilde{X}_A= \frac{X_A}{\sqrt{\rho}}$,  $\tilde{X}_B=h_4\frac{X_B}{\sqrt{\rho}h_3}$, and  $\tilde{Z}_B=\frac{Z_B}{\sqrt{\rho}h_3}$. Note that $I(X_B^n; \tilde{X}_A^{n} + \tilde{X}_B^n+\tilde{Z}_B^n|W_C)\geq I(X_B^n; \tilde{X}_A^{n} + \tilde{X}_B^n+{Z}_B^n|W_C)$ since increasing the noise variance (by $1-\frac{1}{\rho h_3^2}$) leads to a degraded channel, and hence, decreases the mutual information. This leads to $T\geq I(X_B^n; \tilde{X}_A^{n} + \tilde{X}_B^n+Z_B^n|W_C)$. Now, observe that $ I(X_B^n; \tilde{X}_A^{n} + \tilde{X}_B^n+Z_B^n|W_C)$ is larger than $h(\tilde{X}_B^n+Z_B^n|\tilde{X}_A^{n},W_C) - h ( \tilde{X}_A^{n} + Z_B^n|X_B^n,W_C)$ since conditioning reduces entropy. As a result,
\begin{align*}
T &\geq h(\tilde{X}_B^n+Z_B^n|\tilde{X}_A^{n},W_C) - h ( \tilde{X}_A^{n} + Z_B^n|\tilde{X}_B^n,W_C) \\
&= h(\tilde{X}_B^n+Z_B^n|W_C) - h ( \tilde{X}_A^{n} + Z_B^n|W_C),
\end{align*}
since $h(\tilde{X}_B^n+Z_B^n|\tilde{X}_A^{n},W_C)=h(\tilde{X}_B^n+Z_B^n|W_B)=h(\tilde{X}_B^n+Z_B^n|W_A,W_B)$ because $(\tilde{X}_A^n,W_A)$ is independent of $\tilde{X}_B^n$ and $W_B$, and similarly $h(\tilde{X}_A^n+Z_B^n|\tilde{X}_B^{n},W_C)=h(\tilde{X}_A^n+Z_B^n|W_A,W_B)$. Thus,
\begin{align*}
T &\geq  I(X_B^n;\tilde{X}_B^n+Z_B^n|W_C)  - I (\tilde{X}_A^{n}; \tilde{X}_A^{n} + Z_B^n|W_C)\\
&= I(X_B^n; X_B^n+\widehat Z_B^n|W_C)- I (\tilde{X}_A^{n}; \tilde{X}_A^{n} + Z_B^n|W_C)\\
&\geq I(X_B^n; X_B^n+\frac{1}{h_2}Z_B^n|W_C) - I (\tilde{X}_A^{n}; \tilde{X}_A^{n} + Z_B^n|W_C),
\end{align*}
where \textcolor{comm_Ay}{$\widehat Z_B^n = \frac{\sqrt{\rho}h_3}{h_4}Z_B^n$} and the last step follows by increasing the noise variance by $\frac{1}{|h_2|^2} - \frac{\rho|h_3|^2}{|h_4|^2}\geq 0$ (cf. \eqref{eq:ExtCondition3}). Now, we plug in \eqref{StB} to obtain
\begin{align}
&h(Y_A^n|W_C)-h(Y_B^n|W_C)  \nonumber\\
&\leq I(X_B^n;h_2 X_B^n+Z_A^n|W_C) - I(X_B^n; X_B^n+\frac{1}{h_2}Z_B^n|W_C)\nonumber\\
&\quad  + I (\tilde{X}_A^{n}; \tilde{X}_A^{n} + Z_B^n|W_C) \\
&= h (\tilde{X}_A^{n} + Z_B^n|W_C) - h(Z_B^n|W_C) \\
&\leq h ( \tilde{X}_A^{n} + Z_B^n) - h(Z_B^n),
\end{align}
which follows since conditioning reduces entropy and since $Z_B^n$ is independent of $W_C$.  Finally, $h ( \tilde{X}_A^{n} + Z_B^n) - h(Z_B^n)\leq nC(1)=n$ where $C(1)$ is the capacity of a Gaussian channel with input $\tilde{X}_A$ and noise $Z_B^n$, both of unit-power. This concludes the proof of Lemma \ref{Lemma:Entropy_diff_Gaussian}.
\end{proof}

\section{}
\label{app:Constant Gap Analysis for TDMA-TIN}
In this appendix, we transform the upper bound given in Lemma \ref{Th:UB_Det_32X} into a GDoF upper bound. We distinguish between two cases: $\alpha_{j_2i_3}>\alpha_{j_2i_1}$ and $\alpha_{j_2i_3}\leq\alpha_{j_2i_1}$.

\noindent{\bf Case $\alpha_{j_2i_3}>\alpha_{j_2i_1}$:} 
In this case, $d=1$. If $\alpha_{j_2i_1}-\alpha_{j_1i_1}\leq \alpha_{j_2i_3}-\alpha_{j_1i_3}-\alpha_{j_2i_1}$, then $c^2=\rho^{\alpha_{j_2i_1}-\alpha_{j_1i_1}}$. By substituting in $B(\boldsymbol{p})$, we can obtain
\begin{align*}
B(\boldsymbol{p})
&< C\left(\rho^{\alpha_{j_1i_2}} + \frac{\rho^{\alpha_{j_1i_1}}}{\rho^{\alpha_{j_2i_1}}}\right)\notag\\ 
&\quad + C\left(\rho^{\alpha_{j_2i_1}} + \rho^{\alpha_{j_2i_3}} +\frac{\rho^{\alpha_{j_2i_2}} }{\rho^{\alpha_{j_1i_2}}}\right)+1\\
%&< \log_2\left(3\rho^{\max\{\alpha_{j_1i_2},\alpha_{j_1i_1}-\alpha_{j_2i_1}\}}\right)\nonumber\\
%&\quad+ \log_2\left(4\rho^{\max\{\alpha_{j_2i_1},\alpha_{j_2i_3},\alpha_{j_2i_2}-\alpha_{j_1i_2}\}} \right)+1  \notag \\
&< \max\{\alpha_{j_1i_2},\alpha_{j_1i_1}-\alpha_{j_2i_1}\}\log_2(\rho)
\\
&\quad+ \max\{\alpha_{j_2i_1},\alpha_{j_2i_3},\alpha_{j_2i_2}-\alpha_{j_1i_2}\}\log_2(\rho)+5\nonumber.
\end{align*}
For the other case where $\alpha_{j_2i_1}-\alpha_{j_1i_1}> \alpha_{j_2i_3}-\alpha_{j_1i_3}-\alpha_{j_2i_1}$, we have $c^2=\rho^{\alpha_{j_2i_3}-\alpha_{j_2i_1}-\alpha_{j_1i_3}}$. Using  similar steps as above, we can get
\begin{align*}
B(\boldsymbol{p})
&< \max\{\alpha_{j_1i_2},\alpha_{j_1i_3}-(\alpha_{j_2i_3}-\alpha_{j_2i_1})\}\log_2(\rho)
\\
&\quad+ \max\{\alpha_{j_2i_1},\alpha_{j_2i_3},\alpha_{j_2i_2}-\alpha_{j_1i_2}\}\log_2(\rho)+5\nonumber.
\end{align*}
By combining both cases, dividing by $\log(\rho)$, and letting $\rho\to\infty$, we get the following GDoF bound
\begin{align}
&d_{\Sigma,1}(\boldsymbol{p})< \max\{\alpha_{j_2i_1},\alpha_{j_2i_3},\alpha_{j_2i_2}-\alpha_{j_1i_2}\}\label{eq:UB_3x2_appG5}\\
&\quad +\max\{\alpha_{j_1i_2},\alpha_{j_1i_1}-\alpha_{j_2i_1},\alpha_{j_1i_3}-(\alpha_{j_2i_3}-\alpha_{j_2i_1})\}
\nonumber.
\end{align}

\noindent{\bf Case $\alpha_{j_2i_3}\leq \alpha_{j_2i_1}$:}
In this case, $d=0$ and $c^2=\rho^{\alpha_{j_2i_1}-\alpha_{j_1i_1}}$. Similar to the previous case, we can obtain
\begin{align}
\label{eq:UB_3x2_appG6}
d_{\Sigma,2}(\boldsymbol{p})
&< \max\{\alpha_{j_2i_1},\alpha_{j_2i_3},\alpha_{j_2i_2}-\alpha_{j_1i_2}\}\\
&\ +\max\{\alpha_{j_1i_2},\alpha_{j_1i_3},\alpha_{j_1i_1}-\alpha_{j_2i_1}\} 
\nonumber. 
\end{align}

Now, by combining the results in \eqref{eq:UB_3x2_appG5} and \eqref{eq:UB_3x2_appG6}, we obtain
\begin{align*}
&d_\Sigma(\boldsymbol{p})< \max\{\alpha_{j_2i_1},\alpha_{j_2i_3},\alpha_{j_2i_2}-\alpha_{j_1i_2}\}\nonumber\\
&\ \  +\max\{\alpha_{j_1i_2},\alpha_{j_1i_1}-\alpha_{j_2i_1},\alpha_{j_1i_3}-(\alpha_{j_2i_3}-\alpha_{j_2i_1})^+\}
\end{align*} 
and as a result, we get $D(\boldsymbol{p})$ as given in \eqref{GDoFUpperBound}.

%Now, suppose that the conditions in Theorem \ref{Th:TIN_con_gap_32X_Gauss} are satisfied for some permutation of transmitters and receivers $\boldsymbol{p}=(t_1,t_2,t_3,r_1,r_2)$. By using these conditions in \eqref{eq:UB_3x2_appG8} we get 
%\begin{align}
%R_\Sigma&\leq \min D(i_1,i_2,i_3,j_1,j_2)\\
%&\leq B(\boldsymbol{p})\\
%\label{UpperBoundGap}
%&= (\alpha_{r_1t_1}-\alpha_{r_2t_1}+\alpha_{r_2t_2}-\alpha_{r_1t_2})\log_2(\rho)+5.
%\end{align}

\ifCLASSOPTIONcaptionsoff
  \newpage
\fi
\bibliography{bibl}
\end{document}